\documentclass[a4paper,UKenglish,cleveref, autoref, thm-restate]{lipics-v2021}

\pdfoutput=1 %
\hideLIPIcs  %

\usepackage{comment} 

\usepackage{amsfonts}
\usepackage{amssymb}
\usepackage{amsmath}
\usepackage[ruled,vlined]{algorithm2e}
\usepackage[dvipsnames]{xcolor} 
\usepackage{tikz}
\usetikzlibrary{decorations.pathreplacing}

\usepackage{amsthm}

\newcommand{\R}{\mathbb R}
\newcommand{\ignore}[1]{}

\newcommand{\sols}{\mathcal{S}}
\newcommand{\dist}{\mathcal{D}}
\newcommand{\grps}{\mathcal{G}}
\newcommand{\comm}{\mathcal{C}}

\newcommand{\tA}[1]{#1-tradeoff leximax}
\newcommand{\rA}[1]{#1-recursive leximax}

\newcommand{\sig}[1]{#1-significant recursive leximax}

\title{Leximax Approximations and Representative Cohort Selection} %

\author{Monika {Henzinger}}{University of Vienna, Vienna Austria}{mhenz@cs.stanford.edu}{https://orcid.org/0000-0002-5008-6530}{
\flag{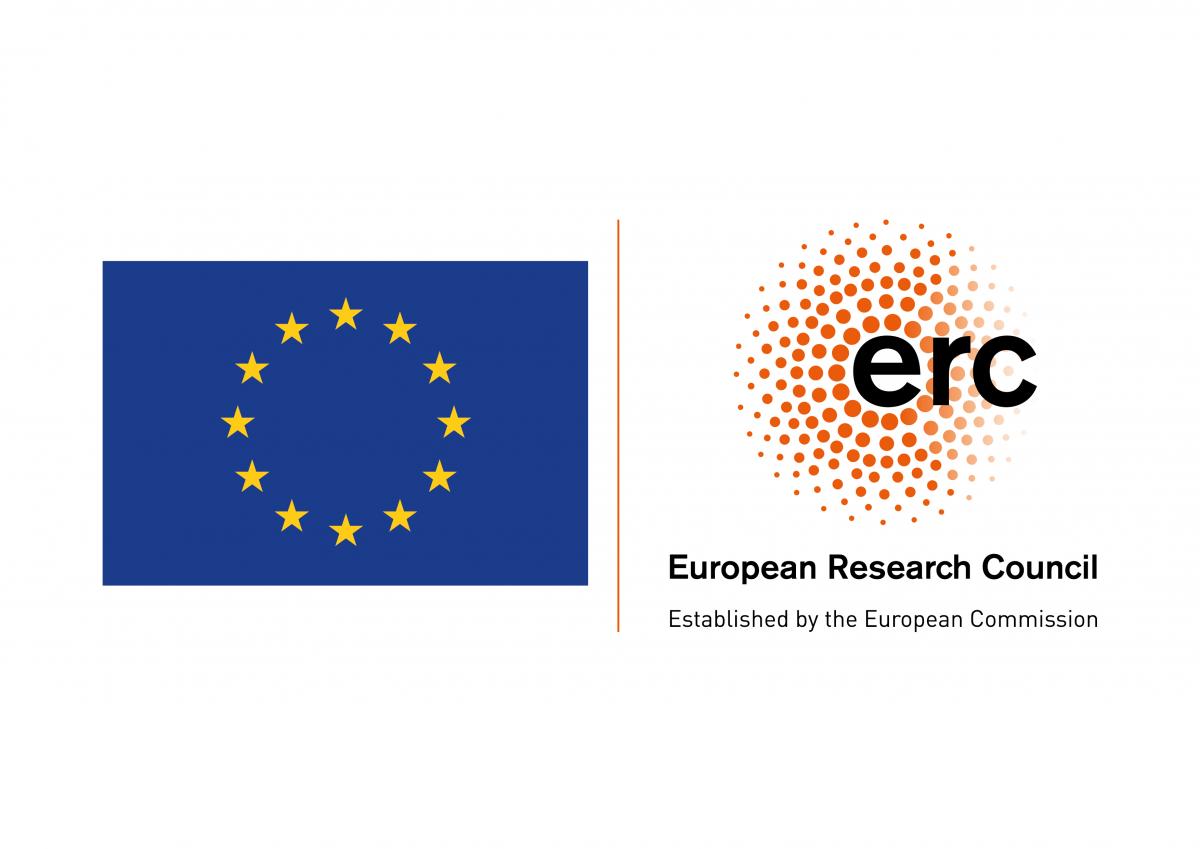}
This work was done in part as Stanford University Distinguished Visiting Austrian Chair. This project has received funding from the
European Research Council (ERC) under the European Union's Horizon 2020
research and innovation programme (Grant agreement No.\ 101019564
``The Design of Modern Fully Dynamic Data Structures (MoDynStruct)''
and from the
Austrian Science Fund (FWF) project ``Fast Algorithms for a Reactive Network
Layer (ReactNet)'', P~33775-N, with additional funding from the \textit{netidee SCIENCE
Stiftung}, 2020--2024.}%

\author{Charlotte {Peale}}{Stanford University, United States}{cpeale@stanford.edu}{https://orcid.org/
0000-0002-9959-857X}{Supported by the Simons Foundation Collaboration on the Theory of Algorithmic Fairness.}%

\author{Omer {Reingold}}{Stanford University, United States}{reingold@cs.stanford.edu}{https://orcid.org/
0000-0003-4997-1716}{Supported by the Simons Foundation Collaboration on the Theory of Algorithmic Fairness, the Sloan Foundation Grant 2020-13941 and the Simons Foundation investigators award 689988.}%

\author{Judy Hanwen {Shen}}{Stanford University, United States}{jhshen@cs.stanford.edu}{https://orcid.org/0000-0002-7864-5242}{Supported by the Simons Foundation Collaboration on the Theory of Algorithmic Fairness.}%

\authorrunning{M. Henzinger et al.} %

\Copyright{Monika Henzinger et al.} %



\category{} %

\relatedversion{} %

\nolinenumbers

\EventEditors{John Q. Open and Joan R. Access}
\EventNoEds{2}
\EventLongTitle{42nd Conference on Very Important Topics (CVIT 2016)}
\EventShortTitle{CVIT 2016}
\EventAcronym{CVIT}
\EventYear{2016}
\EventDate{December 24--27, 2016}
\EventLocation{Little Whinging, United Kingdom}
\EventLogo{}
\SeriesVolume{42}
\ArticleNo{23}
\newcommand{\remove}[1]{}

\begin{document}

\maketitle

\begin{abstract}

Finding a representative cohort from a broad pool of candidates is a goal that arises in many contexts such as choosing governing committees and consumer panels. While there are many ways to define the degree to which a cohort represents a population, a very appealing solution concept is lexicographic maximality (leximax) which offers a natural (pareto-optimal like) interpretation that the utility of no population can be increased without decreasing the utility of a population that is already worse off. However, finding a leximax solution can be highly dependent on small variations in the utility of certain groups. In this work, we explore new notions of approximate leximax solutions with three distinct motivations: better algorithmic efficiency, exploiting significant utility improvements, and robustness to noise. Among other definitional contributions, we give a new notion of an approximate leximax that satisfies a  similarly appealing semantic interpretation and relate it to algorithmically-feasible approximate leximax notions. When group utilities are linear over cohort candidates, we give an efficient polynomial-time algorithm for finding a leximax distribution over cohort candidates in the exact as well as in the approximate setting. Furthermore, we show that finding an integer solution to leximax cohort selection with linear utilities is NP-Hard.

\end{abstract}

\section{Introduction}\label{sec:intro}

In many fairness-related settings, we seek to select an outcome that does not disproportionately harm any key subgroup. Speaking in terms of group utilities, a fair solution would ideally provide every key subgroup with high utility. Unfortunately, such a goal may be impossible to achieve if the utilities derived by subgroups from any potential solutions are in opposition. Moreover, other goals such as seeking to equalize utilities across groups may artificially constrain the utility of certain groups in order to match some group with uniformly low utility.

The classic \emph{maximin} objective, which seeks to output solutions that maximize the utility of the worst-off group, has been widely studied as a goal that can circumvent these potential pitfalls by seeking to achieve the best possible outcome for the worst-off group. This results in a set of solutions that optimize the outcome for the worst-off group, but may still vary quite a bit with respect to the second-worst-off group, third-worst-off group, etc. \emph{Lexicographically maximal} solutions strengthen the maximin objective by requiring that the utility of the second-worst-off group be maximized subject to the worst-off-group achieving its maximin value, the third-worst-off group be maximized subject to the worst-off and second-worst-off values, and so on. This goal intuitively tells us that a lexicographically maximal solution gives the best-possible utility guarantee we can give for each group without harming another group. 

Lexicographic maximality (which we refer to as \emph{leximax}, but is sometimes referred to in the literature as \emph{leximin}) has been widely studied in the context of allocations \cite{freeman2019, kleinberg1999fairness, kurokawaLeximinAllocationsReal2018}. Recently, Diana, Gill, Globus-Harris, Kearns, Roth, and Sharifi-Malvajerdi \cite{dianaLexicographicallyFairLearning2021} explored applying the objective to the contemporary fairness context of loss minimization. In this paper, motivated by the goal of selecting a representative cohort from a group of candidates, we generalize the approach of~\cite{dianaLexicographicallyFairLearning2021} to the goal of selecting a solution that achieves lexicographically maximal utilities for a set of key subgroups.

Our contributions fall into two main categories: definitional, where we explore useful variants of the leximax objective and their relations in the general setting of selecting a leximax solution from a set of potential solutions, and algorithmic, in which we investigate how to efficiently find exact leximax solutions as well as different variants in the specific context of selecting representative cohorts. We provide an overview of definitional contributions in Section~\ref{sec:intro-approx}, followed by an overview of the cohort selection context and resulting algorithms in Section~\ref{sec:intro-algs}.

\subsection{Approximations of Lexicographically Maximal Solutions}\label{sec:intro-approx}

Diana et al.~\cite{dianaLexicographicallyFairLearning2021} define an approximate notion of lexicographic maximality for which they construct oracle-efficient algorithms. Their notion is influenced by an algorithmic approach to calculating leximax solutions in that it assumes the maximal values of the worst-off group, second-worst-off group, etc. are calculated recursively based on whatever estimates came before. The definition assumes some small amount of error when calculating the maximin utility value, and then considers how this error would propagate to the second-worst-off-group's maximum value, then considers how additional errors around the second-worst-off-group's maximum value together with errors from the worst-off group maximum value might propagate to the third-worst-off group, and so on. 

One of the appealing aspects of leximax solutions is that they offer a simple semantic interpretation that explains the sort of fairness guarantees such solutions provide: given a leximax solution, any alternative solution that improves the utility of some group must also decrease the utility of some worse-off group (Proposition~\ref{lem:sem-lex}). While the approximation notion presented by~\cite{dianaLexicographicallyFairLearning2021} is very natural, they also show that such approximate solutions may greatly diverge from exact solutions (see Example~\ref{ex:sensitivity} for details), meaning that they may also diverge from this appealing semantic interpretation. 

Ideally, we'd like a well-defined notion of approximation that extends the semantic interpretation of leximax and relates to the algorithmically achievable notion presented by~\cite{dianaLexicographicallyFairLearning2021}. However, we find that such a definition is somewhat difficult to pin down. \emph{Many natural relaxations of the semantic definition result in notions of approximations where either no solutions are guaranteed to satisfy the notion or the notions themselves may not imply a  meaningful fairness guarantee that is analogue to that offered by leximax solutions. }
Developing a meaningful notion of approximation is exactly the challenge that this paper addresses.

We provide a relaxation of the semantic definition that we term \emph{\tA{$\epsilon$}} (Definition~\ref{def:tradeoff}) that is always guaranteed to exist, and while it is not equivalent to the notion presented in~\cite{dianaLexicographicallyFairLearning2021}, in Theorem~\ref{thm:trade-cg-equiv}, we show that it is equivalent to a stronger variant of their definition that we call\emph{ \rA{$\epsilon$}} (Definition~\ref{def:cg-recursive-approx}). The algorithms of~\cite{dianaLexicographicallyFairLearning2021} have the potential to slightly mis-estimate the maxmin values for different groups, and therefore are only guaranteed to output approximate leximax solutions. The type of mis-estimations that may arise are actually more constrained than the full class of errors their weaker notion of approximation allows for. In particular, solutions outputted by their algorithms actually satisfy our stronger notion of \rA{$\epsilon$}. 

Past explorations of lexicographic maximality have mostly concentrated on finding \emph{exact} leximax solutions. In the design of algorithms, approximations are usually viewed as alternative solutions that are ``almost as good'' as the exact solution and that are computed in settings where it is difficult to efficiently find exact solutions.
In this paper we suggest that in some cases, we may prefer to consider an approximate notion of lexicographic maximality rather than its exact counterpart. In particular, exact leximax solutions may be highly dependent on small variations in the utility of less-well-off groups. For example, a solution where all groups receive 0.01 utility would be preferred by the exact leximax objective over a solution where one group receives 0 utility and all others receive a utility of 1, even though this second solution gets much higher utility for the majority of groups while only decreasing the utility of a single group by a tiny amount. We explore well-defined ways where approximation can benefit stakeholders and suggest a notion of approximation that is stronger than the \rA{$\epsilon$} notion mentioned above that we term \emph{\sig{$\epsilon$} approximation} (Definition~\ref{def:sig-recursive-approx}) that identifies solutions that ignore tiny variations in utility and identifies only solutions that are leximax due to significant increases in utility. In Theorem~\ref{lem:sig-prop}, we give a more formal characterization of the benefits drawn from considering \sig{$\epsilon$} solutions rather than just any \rA{$\epsilon$} solution. 

A third motivator for our study of leximax approximations is how robust leximax solutions may be to small amounts of noise in the estimates of group utility. We show that when calculated in a noisy setting, our relaxed semantic notion (\tA{$\epsilon$}) is not guaranteed to still be \tA{$\epsilon$}, however it is guaranteed to satisfy the weaker notion of approximation defined in~\cite{dianaLexicographicallyFairLearning2021}. On the other hand, in Lemma~\ref{lem:noisy-trade} we show that we can define a stronger variant of the semantic notion that guarantees a solution will be \tA{$\epsilon$} in the noisy setting, but it has the disadvantage that such solutions may not always exist. We also examine noise in the context of \sig{$\epsilon$} solutions, and show that when such solutions are calculated in the presence of noise, they are somewhat robust to noise as they imply a slightly weakened variant of $\epsilon$-significance (Lemma~\ref{lem:sig-w-noise}). 

Figure~\ref{fig:class-diag} summarizes the various notions of approximate lexicographical maximality and how they relate to one another. All of our approximate notions are defined with respect to an arbitrary class of solutions from which we'd like to pick a leximax solution. This allows our new definitions to be applied in the deterministic setting, where each solution would represent a particular cohort, or a randomized setting, where each solution corresponds to a distribution over cohorts and utilities are given in expectation.
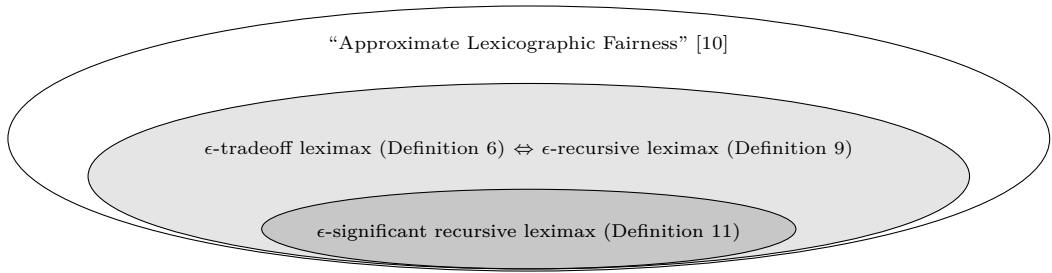
\begin{figure}
    \centering
    \begin{tikzpicture}
        \draw[black] (3,3) ellipse (195pt and 50pt) node[anchor=north, yshift=1.5cm]{\scriptsize``Approximate Lexicographic Fairness''~\cite{dianaLexicographicallyFairLearning2021}};
        \draw[fill=gray, fill opacity=0.2] (3,2.5) ellipse(165pt and 35pt);
        \draw[fill=darkgray, fill opacity=0.2] (3,1.8) ellipse(100pt and 15pt) node[anchor=north, yshift=0.2cm]{\scriptsize \sig{$\epsilon$}  (Definition~\ref{def:sig-recursive-approx})};
        \draw (3, 2.5) node[anchor=north, yshift=0.6cm]{\scriptsize \tA{$\epsilon$} (Definition~\ref{def:tradeoff}) $\Leftrightarrow$ \rA{$\epsilon$} (Definition~\ref{def:cg-recursive-approx})};
        \draw (3,1.8) node[anchor=north, yshift=0.2cm]{\scriptsize \sig{$\epsilon$}  (Definition~\ref{def:sig-recursive-approx})};
    \end{tikzpicture}
    \caption{Relations between the different notions of leximax approximation discussed here and in~\cite{dianaLexicographicallyFairLearning2021}. }
    \label{fig:class-diag}
\end{figure}

\subsection{Algorithms for Leximax Cohort Selection}\label{sec:intro-algs}
In data selection, recruiting, and civic participation settings where a representative cohort is desired, the goal of representation is juxtaposed with the constraint of selecting a small representative set. There can be tension between selecting a cohort small enough for the resources available but large enough to represent as much of the population as possible. A lexicographically maximal solution is particularly salient in a representative cohort problem because it guarantees inclusion for the worst-off-groups while optimizing for the utility of all groups. We consider a model where how well each group or individual in the population is represented by a cohort candidate is given by a utility function. While there are many different ways a cohort or committee in power might make decisions or influence outcomes, we consider a linear setting where the utility a group derives from a cohort is the sum of utilities derived from each member of the cohort. Approximate notions of lexicographical maximality are of particular interest in this setting since estimating utilities that describe representativeness is difficult and might be noisy in practice. 

Diana et al. \cite{dianaLexicographicallyFairLearning2021} give a convex formulation of approximate lexicographical fairness and an oracle-efficient algorithm to solve general leximax convex programs. For our cohort selection setting specifically, we leverage the linearity across decision variables to find a polynomial time algorithm  (Algorithm \ref{alg:lp_cand}) that can calculate both exact leximax solutions as well as the two approximate variants we consider, \tA{$\epsilon$} and \sig{$\epsilon$} (with no external oracle needed for the calculation). 

The linearity of utilities across cohort members and the recursive definition of leximax gives us a sequence of linear programs where the number of variables is linear in the size of the candidate pool and the number of constraints is exponential in the number of groups. In each $m$-th linear program, we maximize the sum of utilities of all sized-$m$ groups which gives us an exponential number of constraints; rendering the linear program too big to solve via generic LP solvers. We circumvent this difficulty by creating a separation oracle (Algorithm \ref{alg:sep-oracle}) which tests the sum of utilities of the $m$ worst off groups efficiently, giving us a polynomial time algorithm overall (Lemma \ref{lem:poly-lp}). We can use the same approach to efficiently find  \tA{$\epsilon$} or \sig{$\epsilon$} solutions by modifying the lower bound constraints on the sum of group utilities. 

The output of our algorithm allows a randomized approach for selecting a cohort of expected size $k$ that guarantees leximax utilities in expectation. We can also round our algorithm output to a solution of size exactly $k$ where the expected utility across groups is leximax\footnote{For further rounding details, see discussion in Section \ref{sec:rounding}}. We focus on this distributional setting for our algorithms for a few key reasons:

\begin{itemize}
    \item \textbf{Tractability.}
    If we wanted to instead find a deterministic cohort of exactly size $k$ by finding a lexicographically maximal integer solution, the problem becomes hard. By showing that the problem of finding the exact integer lexicographically maximal cohort  solves the NP-hard problem of Minimum Hitting Set, we show that finding a solution as well as approximating the number of groups with non-minimum utility within a factor of $(1 - 1/e) + o(1)$ is NP-Hard (Lemma \ref{lem:hardness}).
    
    \item \textbf{Fair Arbitration between Solutions.} It is very possible that two lexicographically maximal deterministic solutions may provide wildly different utility values for a particular group. As an example, consider choosing between a cohort that provides maximum utility to Group A, but zero utility to Group B, and another cohort that provides zero utility to Group A but maximum utility to Group B. Both cohorts are a lexicographic maximum, however selecting a deterministic solution requires us to decide whether the solution should favor Group A or B. A distributional approach gets rid of this difficult decision because the randomized approach itself guarantees that we are providing both A and B a fair chance at high utility. %
\end{itemize}

There are many different potential approaches to randomly selecting a cohort in the distributional setting. We choose to use a randomized approach to selection that includes or excludes each potential cohort member independently with probability outputted by the algorithm. Such an approach offers the following benefits:

\begin{itemize}    
    \item \textbf{Simple Sampling Procedure.} Rather than outputting an arbitrary and potentially complicated distribution over cohorts that is difficult to sample from, the output of our algorithm is a single vector of marginal selection probabilities for each potential candidate. Our approach still results in a cohort with expected size $k$, but provides an easy way to sample cohorts, and as discussed in the final bullet point, gives better guarantees about the utility groups can expect to receive in practice. We also describe a rounding approach that results in cohorts of size exactly $k$ that are still leximax in expectation.
    
    \item \textbf{Better Concentration Guarantees for Some Natural Settings.} While a distributional leximin solution may give groups better utility guarantees in expectation, it comes with the caveat that individual runs of the randomized solution may still result in cohorts where groups receive utility that is far below their expected utility. In an extreme case, a distributional solution that guarantees all groups 0.5 utility might be achieved by choosing uniformly between solutions that provide maximum and zero utility. When the size of the cohort is large enough, our approach to randomized choice guarantees that groups receive utility near their expectation with high probability because we consider each cohort member independently, rather than outputting an arbitrary joint distribution over potential cohort members (Lemma~\ref{lem:concentration}). 

\end{itemize}

\subsection{Our Contributions}\label{sec:contributions}
To summarize, we provide the following contributions:

\begin{enumerate}
    \item Define a new semantic notion of leximax approximation that is always guaranteed to exist and show that it is equivalent to an algorithmically-inspired notion of approximation that is stronger but related to the one defined in~\cite{dianaLexicographicallyFairLearning2021}.
    
    \item Investigate stricter notions of approximation that identify significantly leximax solutions that can be achieved by ignoring small variations in utility. 
    
    \item Explore how our new notions of approximation behave in settings where the group utilities may be reported with some small amount of additive noise.
    
    \item Provide polynomial time algorithms for computing exact and approximate leximax distributions over cohorts with linear utility functions.
    
    \item Show that the alternative goal of computing deterministic cohorts in our setting is NP-hard, and moreover approximating the number of groups with non-minimum utility is also NP-hard.
\end{enumerate}

\subsection{Related Work}\label{sec:related-work}
Fair and diverse selection has become a prominent area of interest in algorithmic and machine learning fairness communities. In the setting of selecting representative data, prior works define metrics for diversity \cite{mitchell2020diversity}, and give algorithms for diverse data selection and summarization \cite{celis2018fair, kleindessner2019fair}. For selecting individuals from a larger pool, prior works on cohort selection and multi-winner elections have studied individual guarantees of fairness \cite{bairaktari2021fair} as well as group parity goals of diversity \cite{bredereck2017multiwinner,  celis2017multiwinner, schumann2017diverse}. Other works have examined how bias and variance may affect different groups differently during a selection process and fairness amounts to remedying implicit bias and variance in the selection process for different groups of individuals \cite{emelianov2020fair, kleinberg2018selection}. Parity or proportional diversity approaches to cohort selection assume the correct amount of representation for each subgroup is known and thus fairness can be achieving a predefined level of diversity. 

When there is no ``merit'' function to guide a selection process, cohort selection can also been seen as a representation problem. Diversity is the goal of a central decision maker while representation is the objective of each group in the population when selecting a cohort. Instead of modeling overall welfare based on the number of representatives from each group, our work considers the welfare of each group based on how representative each cohort member is for that group. Since how well a cohort  serves each group in a population cannot be summarized by a single value, a natural direction is to examine the utilities of all groups of a given cohort that has been selected from a general population. Lexicographical fairness emerges as a reasonable notion of fairness that guarantees Pareto optimality in this setting of multiple objectives or losses. Flanigan et. al. \cite{flanigan2021fair} give an algorithm for recruiting ``citizen's assemblies'' based on sampling from a distribution over representative panels that are generated from leximax selection probabilities over citizens in the population. Our work looks at selecting a representative cohort from a pool of candidates rather than the underlying population which allows a more general model where each member or group in the population has a utility vector describing its utility for each candidate that is being considered for the cohort. Furthermore, we optimize for leximax utilities for each group of interest rather than leximax sample probabilities for each individual in the population.

In telecommunication network design, min-max fairness (MMF) is an important solution concept to lexicographically maximize fractional flow for all parties \cite{ allaloufCentralizedDistributedAlgorithms2008, nace2008max, ogryczak2005telecommunications}. An adjacent problem of lexigraphically maximal flows where there are multiple sinks has also been studied and a polynomial time algorithm exists for finding fractional flow \cite{megiddoOptimalFlowsNetworks1974, megiddoGOODALGORITHMLEXICOGRAPHICALLY1977}. The problem of finding a leximax routing for an unsplittable flow along a network is NP-Complete but finding a 2-approximation is possible \cite{kleinberg1999fairness}. An approximate solution here means that it is not possible to improve a group without decreasing the utility of another group that is more than a factor of 2 worse. 

Lexicographically maximal solutions have also been studied in other domains including bottleneck combinatorial optimization problems \cite{burkard1991lexicographic, della1999improved}, sampling actions for repeated games \cite{balan2008}, allocation of classrooms \cite{kurokawaLeximinAllocationsReal2018} as well as indivisible goods more generally \cite{freeman2019}. It is important to note that unlike the leximax allocation problem, there is no limit on the number of groups gaining utility from the same candidate being included in a cohort or allocated set. Most recently, leximax empirical risk minimization for classification has also been studied \cite{kamani2021pareto, martinezMinimaxPareto2020, martinez21a}. %

\section{The Leximax Objective}

In this paper, we focus on approaches to selecting \emph{lexicographically maximal} (or leximax) representative cohort solutions. We consider a setting in which we'd like to select a solution $S$ from a set of potential solutions $\sols$ such that $S$ is a good representation of some set of key (potentially overlapping) subgroups $\grps = \{G_1, ..., G_m\}$. We measure degree of representation via a utility function $u: \sols \times \grps \rightarrow [0,1]$. 
Ideally, we'd like to select a cohort such that every subgroup is guaranteed to have high utility. However, this may be impossible to achieve in certain settings, such as when the utility functions of two groups are in opposition. Unlike maximizing total welfare, which may result in solutions that neglect the welfare of certain groups or seeking to equalize utilities across groups, which may artificially cap the utility some groups can achieve, lexicographically maximal solutions extend the goal of the classic maxmin objective by seeking to maximize the utility of the worst-off group, and then seeking to maximize the utility of the second-worst-off group subject to this worst-off group's value, etc. This results in a solution concept that seeks to give the best guarantee possible for every key group, rather than just the worst-off. 

We now formally define the leximax objective. 

\begin{definition}
Given two vectors $u$ and $v$ in $\mathbb R^m$, we say that $u$ is lexicographically greater than $v$, or $v \preceq u$, if and only if there exists some $i$ such that for all $j \leq i$ we have $v_j = u_j$, and either $i = m$ or $u_{i + 1} > v_{i + 1}$.
\end{definition}

Applying this definition to the set of sorted group utility vectors obtained from every possible solution gives us a total ordering on these vectors. A leximax solution is any vector that is maximal according to this ordering. In many portions of this paper, in order to reason about the contents of these sorted vectors, we will care about the utility that the $i$th worst-off group receives from a particular solution $S$. We denote this with the bracketed notation $u(S, G_{[i]})$.

\begin{definition}
Given a set of potential solutions $\sols$ and groups $\grps$,
we say that a solution $S \in \sols$ is lexicographically maximal (leximax) if for any other solution $S'$, we have $\langle u(S', G_{[i]})\rangle_{i = 1}^m \preceq \langle u(S, G_{[i]})\rangle_{i = 1}^m$.
\label{def:leximax_def}
\end{definition}

Intuitively, when we seek to find a lexicographically maximal solution, we try to do the best we can for the worst-off group, and then within these potential solutions try to do the best we can for the second-worst-off group, etc. Note that under this definition, groups may achieve varying utilities for different lexicographically maximal solutions, however the vector of sorted group utilities will be unique for any leximax solution. When the solution class is convex and compact and the utility function is continuous with respect to this class, a particular group receives the same utility under any leximax solution.

An attractive feature of lexicographically maximal solutions is that they have an equivalent definition that gives a semantic understanding of the solutions identified by the goal in Definition \ref{def:leximax_def}. We call this notion \emph{tradeoff leximax}.

\begin{proposition}\label{lem:sem-lex}
Given a set of solutions $\sols$ and groups $\grps$, $S \in \sols$ is lexicographically maximal if and only if for any $S'$ and $i \in [m]$ such that $u(S', G_{[i]}) > u(S, G_{[i]})$, there exists some $j < i$ such that $u(S, G_{[j]}) > u(S', G_{[j]})$. 
\end{proposition}

\begin{proof}
In the forward direction, let $S$ be a lexicographically maximal solution. Suppose that we have $i$ and $S'$ such that $u(S, G_{[i]}) < u(S', G_{[i]})$. 

Because $S$ is lexicographically maximal, we know that either $S = S'$ or there exists some $j$ such that $u(S, G_{[j]}) > u(S', G_{[j]})$ and for all $j' < j$, $u(S, G_{[j']}) \geq u(S', G_{[j']})$. 

Because $u(S, G_{[i]}) < u(S', G_{[i]})$, we know that $S \neq S'$, and so such a $j$ must exist, and also $j < i$ otherwise we cannot have $u(S, G_{[j']}) \geq u(S', G_{[j']})$ for all $j' < j$, and therefore the requirements of the statement are met. 

In the opposite direction, suppose we have a solution $S$ such that for any $S'$ and $i \in [m]$ such that $u(S', G_{[i]}) > u(S, G_{[i]})$, there exists some $j < i$ such that $u(S, G_{[j]}) > u(S', G_{[j]})$. 

Let $i$ be the smallest $i \in [m]$ such that $u(S, G_{[i]}) \neq u(S', G_{[i]})$. If no such $i$ exists, then $u(S, G_{[i]}) = u(S', G_{[i]})$ for all $i$ and we trivially have $\langle u(S, G_{[i]})\rangle_{i = 1}^m \succeq \langle u(S', G_{[i]})\rangle_{i = 1}^m$.

Otherwise, suppose for contradiction that $u(S, G_{[i]}) < u(S', G_{[i]})$. By our assumption on $S$, there must exist some $j < i$ such that $u(S, G_{[j]}) > u(S', G_{[j]})$, however this is a contradiction because we have $u(S, G_{[j]}) = u(S', G_{[j]})$ for all $j < i$. Therefore we conclude our assumption was false, and therefore $u(S, G_{[i]}) > u(S', G_{[i]})$, and so $\langle u(S, G_{[i]})\rangle_{i = 1}^m \succeq \langle u(S', G_{[i]})\rangle_{i = 1}^m$. 

Therefore for any other $S'$, we have $\langle u(S, G_{[i]})\rangle_{i = 1}^m \succeq \langle u(S', G_{[i]})\rangle_{i = 1}^m$, and thus $S$ is lexicographically maximal.
\end{proof}

This equivalent definition of lexicographic maximality offers an appealing re-interpretation of this objective: a solution is optimal if increasing the utility of any particular group would result in decreasing the utility of a worse-off group.

\section{Approximations of Leximax-Optimal Solutions}

While the leximax objective's goal of doing the best we can for every group is attractive, one potential downside is that the set of leximax-optimal solutions can be incredibly sensitive to small variations in the utility received by certain groups. We consider the following example that illustrates this phenomenon:

\begin{example}[Sensitivity of leximax-optimal solutions]\label{ex:sensitivity}
Consider a simple setting as in Figure~\ref{fig:exact-lex-sensitivity} in which we have two groups, $\grps = \{G_1, G_2\}$, and would like to decide between two potential solutions $\sols = \{S_1, S_2\}$. The utilities for each group and each solution are defined as $u(S_i, G_j) = U_{ij}$ where $U \in [0,1]^{\sols \times \grps}$ is defined as follows:

$$U = \begin{bmatrix}
0 & 1 \\
0.01 & 0.01
\end{bmatrix}$$
                
Clearly the only leximax solution is $S_2$ (with sorted utility vector $(0.01, 0.01)$), because the worst-off group has value $0.01$ rather than receiving $0$ utility as it does in $S_1$ (which has a sorted utility vector of $(0, 1)$. 

However, if we allow for the possibility that the utility estimates are off by even a tiny amount such as $0.01$, suddenly $S_1$ is also a plausibly leximax solution despite having a completely different value for the second-worst-off group.

\begin{figure}
    \centering
    \begin{tikzpicture}
        \draw (0,0) -- (0,4) node[anchor=south] {$G_{1}$};
        \filldraw[black] (0,0.3) circle (2pt) node[anchor=west]{$S_1$};
        \filldraw[black] (0,1) circle (2pt) node[anchor=west]{$S_2$};
        \draw (2, 0) -- (2, 4) node[anchor=south] {$G_{2}$};
        \filldraw[black] (2,1) circle (2pt) node[anchor=west]{$S_2$};
        \filldraw[black] (2,3.5) circle (2pt) node[anchor=west]{$S_1$};
        \draw [decorate,decoration={brace,amplitude=5pt},xshift=0,yshift=0pt]
(-0.2,0.3) -- (-0.2,1) node [black,midway,xshift=-0.6cm]
{\footnotesize $0.01$};
    \draw [decorate,decoration={brace,amplitude=5pt, mirror},xshift=0,yshift=0pt]
(2.1,1.2) -- (2.1,3.3) node [black,midway,xshift=0.6cm]
{\footnotesize $0.99$};

    \end{tikzpicture}
    \caption{Visual representation of the setting in Example~\ref{ex:sensitivity} showing how exact leximax solutions are very sensitive to small changes in utility for less-well-off groups.}
    \label{fig:exact-lex-sensitivity}
\end{figure}
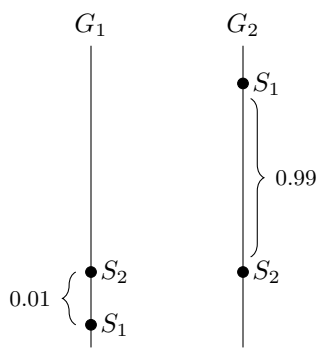
\end{example}

Example~\ref{ex:sensitivity} is notable in that it demonstrates how small variations in the utilities of groups can lead to drastic changes with respect to the types of leximax solutions that are considered optimal. In settings where utilities may be reported with some estimation error, it is therefore incredibly important to consider how these errors might affect how the output optimal solution compares to the true leximax solution that would have been produced given completely accurate utilities. 

Moreover, even when the utilities are believed to be accurate, it may be useful to consider solutions that are not exactly leximax, but are leximax when small variations in the utility are ignored. Example~\ref{ex:sensitivity} is a situation where the exact leximax offers a tiny improvement in the worst-off group at the cost of a huge decrease in the utility of the second-worst-off group. %
A practitioner who views utility differences of less than 0.05 as insignificant might prefer %
$S_1$ as the only \emph{significantly} lexicographically maximal solution because the worst-off groups between $S_2$ and $S_1$ receive comparable utility while the second-worst-off group is significantly better off under $S_1$. 

The search for plausibly exact lexicographic solutions given the potential for some amount of estimation error as well as the need for significantly maximal lexicographic solutions even when working with exact utility values motivates our study of new approximate leximax notions. In this section, we introduce two such notions: first, we introduce a semantic notion of approximate leximax that relaxes the standard leximax definition to consider additional solutions that may be plausibly leximax. The second notion we introduce here seeks solutions that are leximax if only ``significant" improvements are considered (as in the discussion above). Unlike the first notion, the notion of significantly leximax solutions is not a strict relaxation of leximax and may not include the exact leximax solution in some cases.  

\subsection{Relaxations of the Leximax Objective}
\subsubsection{Elementwise Approximation}
The most naive approach to approximation would be to require that the element-wise distance between the sorted utility vectors of the true lexicographically maximal solution and the approximate solution be small: %

\begin{definition}[Element-wise leximax approximation]\label{def:elementwise}
Given a set of $m$ groups $\grps$ and a set of potential solutions $\sols$, let $\ell$ be the sorted vector of utilities attained by any leximax solution. We say that a solution $S \in \sols$ is an $\alpha$-element-wise leximax approximation iff $\max_{i \in [m]}\{\ell_i - u(S, G_{[i]})\} \leq \alpha$.
\end{definition}

While attractive in its simplicity, \cite{dianaLexicographicallyFairLearning2021} observe that in certain contexts, such a definition may be stricter than we can hope for. In particular, if the leximax solution is being computed recursively, small estimation errors in the values of the worst-off group's utility can greatly effect the difference between the utility of better-off groups in a lexicographically maximal solution compared to a solution that maximizes group utilities based off of this incorrect value. Thus, we turn our attention to weaker notions of approximation.

\subsubsection{Tradeoff Approximation}

We introduce a new notion of approximation that is a natural relaxation of the semantic interpretation of leximax solutions provided by the \emph{tradeoff leximax} objective discussed in Proposition~\ref{lem:sem-lex}.

\begin{definition}[\tA{$\epsilon$}]\label{def:tradeoff}
Given a set of $m$ groups, $\grps$, and a set of potential solutions, $\sols$, a solution $S \in \sols$ is \tA{$\epsilon$} if for any $S'$ and $i$ such that 
$u(S, G_{[i]}) < u(S', G_{[i]}) - \epsilon$, there exists a $j < i$ such that 
$u(S, G_{[j]}) > u(S', G_{[j]})$. 
\end{definition}

Intuitively, this definition guarantees that if we can find some other solution that does a lot better on some particular group, then this new solution must also decrease the utility of some worse-off group. 

\tA{$\epsilon$} provides an appealingly simple relaxation of the semantic interpretation of exact leximax solutions. However,
slight variations of this definition, also natural relaxations of leximax, will result in definitions where solutions are not guaranteed to exist. We explore this in the following example:

\begin{example}[Altered versions of \tA{$\epsilon$} may not have any solutions.]\label{ex:sig-trade-problems}
We define a class of alternative tradeoff definitions that we term $(\epsilon_1, \epsilon_2)$-significant tradeoff leximax for reasons that will become clear in Section~\ref{sec:significant-sols} as follows:

\begin{definition}[$(\epsilon_1, \epsilon_2)$-significant tradeoff leximax]\label{def:sig-trade}
Given a set of $m$ groups, $\grps$, and a set of potential solutions, $\sols$, a solution $S \in \sols$ is $(\epsilon_1, \epsilon_2)$-significant tradeoff leximax for any $\epsilon_1, \epsilon_2 \geq 0$ if for any $S'$ and $i$ such that 
$u(S, G_{[i]}) < u(S', G_{[i]}) - \epsilon_1$, there exists a $j < i$ such that 
$u(S, G_{[j]}) > u(S', G_{[j]}) + \epsilon_2$. 
\end{definition}

When $\epsilon_1 = \epsilon$ and $\epsilon_2 = 0$, this notion is equivalent to \tA{$\epsilon$}. When $\epsilon_2 > 0$, the definition requires that any increase by more than $\epsilon_1$ result in a decrease of more than $\epsilon_2$ in a worse-off group. 

However, we demonstrate that for $\epsilon_1, \epsilon_2 > 0$, no solution may exist. Consider $\epsilon_1 = \epsilon_2 = \epsilon$ for the setting depicted in Figure~\ref{fig:sig-trade-diag} where we have two groups and four potential solutions with utilities defined as $u(S_i, G_j) = U_{ij}$, for 
$$U = \begin{bmatrix}
0 & 0.5 + 6\epsilon\\
\epsilon/2 & 0.5 + 4\epsilon\\
\epsilon & 0.5 +2\epsilon \\
3\epsilon/2 & 0.5
\end{bmatrix}$$

Where we assume $\epsilon$ is sufficiently smaller than 0.5. Under these utilities, $S_4$ cannot be $(\epsilon, \epsilon)$-significant tradeoff leximax because $S_3$ improves by more than $\epsilon$ in $G_2$ while only decreasing $G_1$ by $\epsilon/2$. Similarly, $S_3$ and $S_2$ cannot be $(\epsilon, \epsilon)$-significant tradeoff leximax due to the existence of $S_2$ and $S_1$, respectively. This means that $S_2, S_3, S_4$ all cannot be $(\epsilon, \epsilon)$-significant tradeoff leximax. However, we see that $S_4$ improves by more than $\epsilon$ over $S_1$ in $G_1$, so $S_1$ also cannot be  $(\epsilon, \epsilon)$-significant tradeoff leximax. We conclude that no potential solution satisfies this definition\footnote{This example was not tied to the specific choice of $\epsilon_1 = \epsilon_2 = \epsilon$. Similar examples exist for other choices.}.

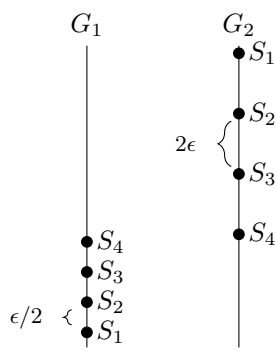
\begin{figure}
    \centering
    \begin{tikzpicture}
        \draw (0,0) -- (0,4) node[anchor=south] {$G_{1}$};
        \filldraw[black] (0,0.2) circle (2pt) node[anchor=west]{$S_1$};
        \filldraw[black] (0,0.6) circle (2pt) node[anchor=west]{$S_2$};
        \filldraw[black] (0,1.0) circle (2pt) node[anchor=west]{$S_3$};
        \filldraw[black] (0,1.4) circle (2pt) node[anchor=west]{$S_4$};
        \draw (2, 0) -- (2, 4) node[anchor=south] {$G_{2}$};
        \filldraw[black] (2,1.5) circle (2pt) node[anchor=west]{$S_4$};
        \filldraw[black] (2,2.3) circle (2pt) node[anchor=west]{$S_3$};
        \filldraw[black] (2,3.1) circle (2pt) node[anchor=west]{$S_2$};
        \filldraw[black] (2,3.9) circle (2pt) node[anchor=west]{$S_1$};
        \draw [decorate,decoration={brace,amplitude=3pt},xshift=0,yshift=0pt]
(-0.2,0.3) -- (-0.2,0.5) node [black,midway,xshift=-0.6cm]
{\footnotesize $\epsilon/2$};
   \draw [decorate,decoration={brace,amplitude=5pt},xshift=0,yshift=0pt]
(1.9, 2.4) -- (1.9,3) node [black,midway,xshift=-0.6cm]{\footnotesize $2\epsilon$};

    \end{tikzpicture}
    \caption{Visual representation of the setting in Example~\ref{ex:sig-trade-problems} demonstrating how no solutions may exist under small alterations to the definition of \tA{$\epsilon$}. }
    \label{fig:sig-trade-diag}
\end{figure}

\end{example}

\paragraph{Computing tradeoff approximations recursively}

The definition of \tA{$\epsilon$} is only useful if we can compute \tA{$\epsilon$} solutions efficiently. To show that this is possible, we relate \tA{$\epsilon$} to a different notion of leximax approximation that arises from a natural algorithmic approach and is closely related to the notion of leximax approximations introduced in \cite{dianaLexicographicallyFairLearning2021}. 

Consider the following approach to computing an exact leximax solution, which follows its definition: Compute the maximum value that can be guaranteed to the worst-off group, then calculate the maximum value that can be guaranteed to the second worst-off group subject to this value, and then recurse on the third, fourth, fifth, etc. until the values for all $m$ groups are fixed and a solution is found. 

However, what if our algorithm is not completely accurate at each step? Introducing some amount of estimation error at each step of the recursion may result in selecting a solution that isn't exact leximax, but can considered approximately leximax because it arose from small estimation errors in our algorithm. We call such solutions \rA{$\epsilon$}, and define them as follows:

\begin{definition}[\rA{$\epsilon$}]\label{def:cg-recursive-approx} 
Given a set of $m$ groups, $\grps$, a set of potential solutions $\sols$, and a choice of allowable `slack' $\vec{\alpha} = (\alpha_1, ..., \alpha_m)$ with $\alpha_i \in \mathbb R_{\geq 0}$, recursively define the sets of solutions $\sols^{\alpha}_0, ..., \sols^{\alpha}_m \subseteq \sols$ such that $\sols^{\alpha}_0 := \sols$ and for each $i = 1, ..., m$, 
$$\sols^{\alpha}_i = \{ S \in \sols^{\alpha}_{i - 1} : u(S, G_{[i]}) \geq \max_{S' \in \sols^{\alpha}_{i - 1}} u(S', G_{[i]}) - \alpha_i\}$$

We say that $S \in \sols$ is an $\epsilon$-recursively approximate leximax solution if there exists an $\vec{\alpha}$ with $\max_{i \in [m]}\alpha_i \leq \epsilon$ such that $S \in \sols_m^{\alpha}$.
\end{definition}

Our definition of \rA{$\epsilon$} is a stronger variant of the definition of approximation used in \cite{dianaLexicographicallyFairLearning2021}. Most importantly, the definition presented in \cite{dianaLexicographicallyFairLearning2021} is less strict because it allows for the choice of allowable slack to depend on each solution. However, the solutions outputted by their algorithms actually achieve the stronger notion presented here. Unlike the weaker version, which is only implied by \tA{$\epsilon$}, we can show that \rA{$\epsilon$} and \tA{$\epsilon$} are equivalent. 

In this definition, the choice of slack, $\vec{\alpha} \in [0, \epsilon]^m$, determines the amount of estimation error at each step. We use this $\vec{\alpha}$ to recursively construct the sets $\sols_i^{\alpha}$ in the same way they would be calculated had we applied a recursive approach to calculating a leximax solution but under-estimated the maximum value by $\alpha_i$ at the $i$th step for each $i = 1, ..., m$.

Unlike our \tA{$\epsilon$} notion of approximation, \rA{$\epsilon$} provides a natural algorithmic interpretation of approximate solutions which allows efficient approaches to computing \rA{$\epsilon$} solutions with respect to a particular choice of slack, as we do in Section~\ref{sec:lp-approaches}\footnote{\cite{dianaLexicographicallyFairLearning2021} give algorithms that calculate \rA{$\epsilon$} solutions because their approach estimates each sequential maxmin value to within $\epsilon$ of its true value, though the notion of efficiency that they achieve does not exactly correspond to polynomial-time algorithms. We provide an alternative polynomial-time algorithm for the cohort selection setting that leverages linear group utilities to offer a more efficient approach.}. Fortunately, we can actually show that these two notions of approximation are equivalent, which means that we can also efficiently compute \tA{$\epsilon$} solutions. 

\begin{theorem}\label{thm:trade-cg-equiv}
For any set of groups, $\grps$, and solutions, $\sols$, the set of \tA{$\epsilon$} solutions is equivalent to the set of \rA{$\epsilon$} solutions.
\end{theorem}

\begin{proof}
First, suppose we have some \tA{$\epsilon$} solution $S$. 

Recursively define an amount of allowable slack $\vec{\alpha} \in \mathbb R^m$ as follows, where $\sols_0^{\alpha}, ..., \sols_m^{\alpha}$ are the recursively defined sets discussed in Definition~\ref{def:cg-recursive-approx}:

$$\alpha_{i} = \max_{S' \in \sols_{i - 1}^{\alpha}}u(S', G_{[i]}) - u(S, G_{[i]})$$

In other words, $\alpha_i$ is exactly the distance from the utility of the $i$th worst-off group for $S$ to the maximal utility achieved by any $i$th worst-off group in the $i - 1$th recursive set. 

Clearly under this choice of slack, $S \in \sols_m^{\alpha}$. If $\max_{i \in [m]} \alpha_i \leq \epsilon$, then $S$ is also \rA{$\epsilon$} and we are done. Otherwise, assume for contradiction that this is not the case, and let $\alpha_i$ be the smallest $i$ such that $\alpha_i > \epsilon$. By definition, this means we have some $S' \in \sols_{i - 1}^{\alpha}$ such that $u(S', G_{[i]}) > u(S, G_{[i]}) + \epsilon$.

Moreover, by definition of our $\alpha$, we also know that for all $i' \leq i$, we have
\begin{align*}
    u(S', G_{[i']}) &\geq \max_{S'' \in \sols_{i' - 1}^{\alpha}}u(S'', G_{[i']}) - \alpha_{i'} \\
    &= \max_{S'' \in\sols_{i' - 1}^{\alpha}}u(S'', G_{[i']}) - \left( \max_{S' \in \sols_{i' - 1}^{\alpha}}u(S', G_{[i']}) - u(S, G_{[i']})\right)\\
    &= u(S, G_{[i']})
\end{align*}
and therefore $u(S', G_{[i']}) \geq u(S, G_{[i']})$.

However, because $S$ is \tA{$\epsilon$}, we must also have some $i' < i$ such that $u(S', G_{[i']}) < u(S, G_{[i']})$. This is a contradiction, and so we conclude that $\alpha_i \leq \epsilon$ for all $i \in [m]$, and therefore $S$ is also \rA{$\epsilon$}. 

In the other direction, suppose that $S$ is \rA{$\epsilon$} with respect to some choice of allowable slack $\vec{\alpha}\in\mathbb{R}^m$. We define a new choice of slack $\vec{\alpha'}:[m] \times \sols \rightarrow \mathbb R_{\geq 0}$ as follows:
$$\alpha'_i(S') = \begin{cases}\epsilon & S' = S \\ 0 & \text{otherwise} \end{cases}$$

Consider any $S' \in \sols$ and $i \in [m]$ such that $$u(S', G_{[i]}) > u(S, G_{[i]}) + \alpha'_i(S) = u(S, G_{[i]}) + \epsilon.$$

Because $S$ is \rA{$\epsilon$}, we therefore must have $S' \not\in \sols_i^{\alpha}$ to avoid a contradiction. Let $j$ be the smallest $j$ such that $S'\not\in \sols_j^{\alpha}$. Here, we are guaranteed that 
$$u(S', G_{[j]}) < \max_{S'' \in \sols_{j}^{\alpha}} u(S'', G_{[j]}) - \alpha_j \leq u(S, G_{[j]})$$

Where the left-hand inequality arises because $S'$ must have been too far below the maximum at $j$ because it was eliminated, and the right-hand side is because we know that $S \in \sols_j^{\alpha}$. Thus, 
$$u(S', G_{[j]}) < u(S, G_{[j]})$$
and so because we know that $\alpha'_j(S') = 0$, we have found a $j < i$ such that 
$$u(S', G_{[j]}) < u(S, G_{[j]}) - \alpha_j(S')$$

and therefore $S$ must also be \tA{$\epsilon$}. 

We have shown that any \tA{$\epsilon$} solution must also be \rA{$\epsilon$} and vice versa, so we conclude that the two notions are equivalent. 
\end{proof}

\subsection{Significantly Leximax Solutions}\label{sec:significant-sols}

\tA{$\epsilon$} solutions are strict relaxations of the exact leximax objective. Any leximax-optimal solution will also be \tA{$\epsilon$} and will also be \rA{$\epsilon$} for any $\epsilon \geq 0$ (by simply selecting the allowable slack to be $\alpha_i = 0$ for all $i \in [m]$). Similarly, any $\epsilon$-tradeoff (resp. recursively) approximate solution will also be $\epsilon'$-tradeoff (recursively) approximate for any $\epsilon' \geq \epsilon$. 

In this section, we introduce a modified notion of \rA{$\epsilon$} that is not a relaxation of the exact leximax objective but rather tries to get significant improvements in the quality of solutions, using the allowed slack. This notion  constrains the choices of slack so that solutions considered leximax due to only insignificant improvements in the utility of worse-off groups are ignored. Here, the only slack considered is where all allowable slack values are set to exactly $\epsilon$, rather than some value that is at most $\epsilon$.

\begin{definition}[\sig{$\epsilon$}]\label{def:sig-recursive-approx} 
Given a set of groups $\grps$ with $|\grps| = m$ and a set of potential solutions $\sols$, recursively define the sets of solutions $\sols^{\epsilon}_0, ..., \sols^{\epsilon}_m \subseteq \sols$ such that $\sols^{\epsilon}_0 := \sols$ and for each $i = 1, ..., m$, 
$$\sols^{\epsilon}_i = \{ S \in \sols^{\epsilon}_{i - 1} : u(S, G_{[i]}) \geq \max_{S' \in \sols^{\epsilon}_{i - 1}} u(S', G_{[i]}) - \epsilon\}$$

We say that $S \in \sols$ is \sig{$\epsilon$} if $S \in \sols_m^{\epsilon}$.
\end{definition}

Why does this make sense as a way to identify significant solutions? Intuitively, setting every slack value to the maximum possible $\epsilon$ requires that the valid solutions be leximax with respect to the larger set of potential solutions when some error term is allowed, rather than putting a lot of weight on small differences in earlier groups. We present the following example to see this in practice:

\begin{example}[Significantly recursive approximations]\label{ex:sig-rec}
Consider two groups and two solutions as in Figure~\ref{fig:sig-rec-diag} with utilities
$$u(S_1, G_{1}) = \epsilon, u(S_2, G_{1}) = 0 , u(S_1, G_{2}) = 0.5, u(S_2, G_{2}) = 1.$$

Both $S_1$ and $S_2$ are \rA{$\epsilon$} approximations. If we set $\alpha_1 < \epsilon$, then $S_1$ because the only acceptable solution and thus an \rA{$\epsilon$}-approximate solution. If we set $\alpha_1 = \epsilon$, $S_2$ becomes an \rA{$\epsilon$}-approximate solution. 
We would expect a satisfying significant approximation notion to identify $S_2$ as the only $\epsilon$-significant approximation because it's not too far below $S_1$ on the worst-off group, but does much better on the second-worst-off group. An \sig{$\epsilon$} approximation does give us this separation between $S_1$ and $S_2$, because while both $S_1$ and $S_2$ are included in the first-level of recursion, $\sols_1^{\epsilon}$, $S_1$ is too far below the maximum to be included in $\sols_2^{\epsilon}$, so $S_2$ is the only \sig{$\epsilon$} approximation in this example. 

\begin{figure}
    \centering
    \begin{tikzpicture}
        \draw (0,0) -- (0,4) node[anchor=south] {$G_{1}$};
        \filldraw[black] (0,0.5) circle (2pt) node[anchor=west]{$S_2$};
        \filldraw[black] (0,1.3) circle (2pt) node[anchor=west]{$S_1$};
        \draw (2, 0) -- (2, 4) node[anchor=south] {$G_{2}$};
        \filldraw[black] (2,1.5) circle (2pt) node[anchor=west]{$S_1$};
        \filldraw[black] (2,3.5) circle (2pt) node[anchor=west]{$S_2$};
        \draw [decorate,decoration={brace,amplitude=5pt},xshift=0,yshift=0pt]
(-0.2,0.5) -- (-0.2,1.3) node [black,midway,xshift=-0.6cm]
{\footnotesize $\leq \epsilon$};
    \draw [decorate,decoration={brace,amplitude=5pt, mirror},xshift=0,yshift=0pt]
(2.1,1.7) -- (2.1,3.3) node [black,midway,xshift=0.6cm]
{\footnotesize $> \epsilon$};

    \end{tikzpicture}
    \caption{Visual representation of the setting in Example~\ref{ex:sig-rec} demonstrating how Definition~\ref{def:sig-recursive-approx} identifies significantly leximax solutions. }
    \label{fig:sig-rec-diag}
\end{figure}
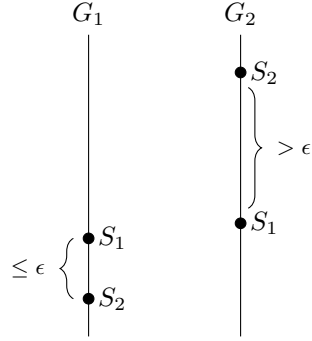
\end{example}

So far, we have been rather loose in arguing about why the solutions identified as \sig{$\epsilon$} might be preferred over exact leximax or the more general class of \rA{$\epsilon$} solutions. We offer a more formal characterization here, but begin by taking a step back to reframe what the contents of the recursively defined sets from Definition~\ref{def:cg-recursive-approx}, $\sols_1^{\alpha}, ..., \sols_m^{\alpha}$ for some choice of slack $\vec{\alpha}$, can tell us about potential leximax solutions.

Intuitively, $\sols_i^{\alpha}$ contains all solutions that, with respect to the first $i$ groups, could feasibly be solutions that are \rA{$\epsilon$} allowing for a slack of $\vec{\alpha}$, and are guaranteed to be within $\epsilon$ of the first $i$ coordinates of any final \rA{$\epsilon$} solution with respect to $\vec{\alpha}$, i.e. any $S \in \sols_m^{\alpha}$.

This means that looking at the maximum utility achieved by any solution in each recursive group, $\langle \max_{S \in \sols_i^{\alpha}}u(S, G_{[i]})\rangle_{i = 1}^m$ gives us a sense of the type of solution that results from allowing $\vec{\alpha}$ as slack. While there may not exist a $S' \in \sols_m^{\alpha}$ such that $\langle u(S', G_{[i]}) \rangle_{i = 1}^{m} = \langle \max_{S \in \sols_i^{\alpha}}u(S, G_{[i]})\rangle_{i = 1}^m$, we are guaranteed that any $S' \in \sols_m^{\alpha}$ will be elementwise within $\epsilon$ of this vector of maximums. 

We can show that out of all possible choices of slack, the one used by the definition of \sig{$\epsilon$}, $\vec{\alpha} = (\epsilon, ..., \epsilon)$ results in the best-possible sequence of maximum set values (i.e. it will be lexicographically greater than the maximums attained via any other choice of slack). In other words, this backs up the motivation behind our definition of \sig{$\epsilon$} in that it promises us that any \sig{$\epsilon$} solution will be elementwise within $\epsilon$ of the lexicographically best solution we could possibly hope for under an optimal choice of slack. 

\begin{theorem}[Leximax properties of \sig{$\epsilon$}]\label{lem:sig-prop}
Given a set of groups, $\grps$, and solutions, $\sols$, let $\sols_1^{\epsilon}, ..., \sols_m^{\epsilon}$ be the recursively defined sets constructed with a slack of $\epsilon$ at each step, as used in the definition of \sig{$\epsilon$}, and for any $\vec{\alpha} \in \mathbb{R}_{\geq 0}^m$, let $\sols_1^{\alpha}, ..., \sols_m^{\alpha}$ be the sets that arise when the amount of allowable slack at each level is set according to $\vec{\alpha}$. Then, for any $\vec{\alpha} \in \mathbb{R}_{\geq 0 }^m$, we have 
$$\langle\max_{S \in \sols_i^{\epsilon}}u(S, G_{[i]})\rangle_{i = 1}^m \succeq \langle\max_{S \in \sols_i^{\alpha}}u(S, G_{[i]})\rangle_{i = 1}^m $$

In other words, the vector of maximums attained in each $\sols_i^{\epsilon}$ is lexicographically maximal compared to any other choice of slack of size at most $\epsilon$.
\end{theorem}

\begin{proof}
Let $\vec{\alpha} \in \mathbb{R}_{\geq 0}^m$. We proceed by induction on $i = 1, ..., m$.

As our base case, we note that, 

$$\max_{S \in \sols_1^{\alpha}}u(S, G_{[1]}) = \max_{S \in \sols}u(S, G_{[1]})$$

and therefore $\max_{S \in \sols_1^{\alpha}}u(S, G_{[1]}) = \max_{S \in \sols}u(S, G_{[1]}) = \max_{S \in \sols_1^{\epsilon}}u(S, G_{[1]})$.

For the recursive case, assume that for all $j < i$, we have $\max_{S \in \sols_j^{\epsilon}}u(S, G_{[j]}) = \max_{S \in \sols_j^{\alpha}}u(S, G_{[j]})$.

For any $S \in \sols_{i}^{\alpha}$, we are guaranteed that for all $j < i$, 
$$u(S, G_{[j]}) \geq \max_{S' \in \sols_j^{\alpha}}u(S', G_{[j]}) - \alpha_j \geq \max_{S' \in \sols_j^{\alpha}}u(S', G_{[j]}) - \epsilon$$

and therefore, by our inductive assumption, 
$$u(S, G_{[j]}) \geq \max_{S' \in \sols_j^{\epsilon}}u(S', G_{[j]}) - \epsilon$$
and so $S \in \sols_i^{\epsilon}$ as well, and therefore because $\sols_i^{\alpha} \subseteq \sols_i^{\epsilon}$, we must have $\max_{S \in \sols_i^{\epsilon}}u(S, G_{[i]}) \geq \max_{S \in \sols_i^{\alpha}}u(S, G_{[i]})$.

Therefore, we've shown that either $\max_{S \in \sols_i^{\epsilon}}u(S, G_{[i]}) = \max_{S \in \sols_i^{\alpha}}u(S, G_{[i]})$ for all $i = 1, ..., m$, or there exists some $i$ such that  $\max_{S \in \sols_i^{\epsilon}}u(S, G_{[i]}) > \max_{S \in \sols_i^{\alpha}}u(S, G_{[i]})$ and $\max_{S \in \sols_j^{\epsilon}}u(S, G_{[j]}) = \max_{S \in \sols_j^{\alpha}}u(S, G_{[j]})$ for all $j < i$, and so the vector of maximums attained by setting the allowable slack to be $\epsilon$ at all levels is lexicographically maximal. 
\end{proof}

Theorem~\ref{lem:sig-prop} tells us that out of all the ways we could identify approximate leximax solutions that ignore variations of less than $\epsilon$, an $\epsilon$-significant solution is guaranteed to be element-wise within $\epsilon$ on the lexicographically maximal best-possible guarantee we can give for each group at each level of recursion. 

Ideally, we could obtain a similar notion to \sig{$\epsilon$} with a satisfying semantic meaning as for \rA{$\epsilon$} by modifying our definition of \tA{$\epsilon$} so that any solution that improves the $i$th group by more than $\epsilon$ must also decrease some worse-off group by more than $\epsilon$. However, as we saw in Example~\ref{ex:sig-trade-problems}, modifying the original tradeoff definition in this way surprisingly results in an overly strict notion due to some instability arising from the pairwise comparisons that tradeoff approximations rely on. In particular, solutions that satisfy this notion may not exist. Note that in Example~\ref{ex:sig-trade-problems}, no solution satisfied $(\epsilon, \epsilon)$-significant tradeoff leximax, which is equivalent to the modified definition suggested here, but $S_2$ is an \sig{$\epsilon$} approximation and $S_2, S_3, S_4$ are all valid \rA{$\epsilon$} solutions.

\subsection{Approximations in the Presence of Noise}\label{sec:noisy-approx}

So far, we have considered approximate leximax solutions with the assumption that the utilities used to calculate these solutions are known to be correct. However, a natural question is how such approximations behave if the reported utilities contain some small amount of noise. 

In the case of \tA{$\epsilon$} solutions, assuming a small amount of additive noise for each utility has the potential for resulting in solutions that do not satisfy tradeoff guarantees. In particular, noise that is solution-specific can cause individual solutions to be ``kicked out'' of the recursively defined sets, even though all solutions near them are included. We demonstrate this behavior in the following example:

\begin{example}[\tA{$\epsilon$} solutions are not robust to noise.]\label{ex:rec-not-trade}

We consider a setting in which we have two groups, $\grps = \{G_1, G_2\}$ and three potential solutions $\sols = \{S_1, S_2, S_3\}$. The utilities each group derives are defined as $u(S_i, G_j) = U_{ij}$ where $U$ is defined as follows (assume $\epsilon << 0.1$):

$$U = \begin{bmatrix}
0.1 & 0.2\\
0.1 + \epsilon/100 & 0.8\\
0.1 + \epsilon & 0.2 
\end{bmatrix}$$

Furthermore, assume we have a slightly noisy version of utilities in which $u(S_2, G_1)$ changes from $0.1 + \epsilon/100$ to $0.1 - \epsilon/100$. Figure~\ref{fig:rec-not-trade-diag} provides a visual representation of this instance, where the noisy verison of $S_2$ is shown in red. 

In the non-noisy version, $S_1$ can never be considered \tA{$\epsilon$} because $S_2$ does much better than $S_1$ on $G_2$, and is still above $S_1$ on $G_1$. 

However, in the noisy version, which introduces only a tiny amount of noise ($\epsilon/50$), much smaller than the allowed approximation threshold ($\epsilon$), results in a setting where $S_1$ can be considered \tA{$\epsilon$}. 

\end{example}

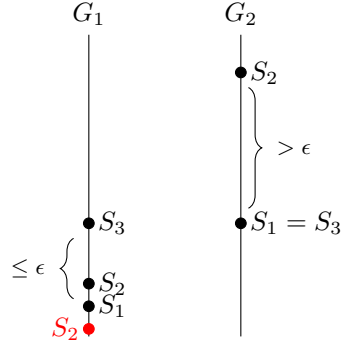
\begin{figure}
    \centering
    \begin{tikzpicture}
        \draw (0,0) -- (0,4) node[anchor=south] {$G_{1}$};
        \filldraw[black] (0,0.4) circle (2pt) node[anchor=west]{$S_1$};
        \filldraw[black] (0,0.7) circle (2pt) node[anchor=west]{$S_2$};
        \filldraw[red] (0,0.1) circle (2pt) node[anchor=east]{$S_2$};
        \filldraw[black] (0,1.5) circle (2pt) node[anchor=west]{$S_3$};
        \draw (2, 0) -- (2, 4) node[anchor=south] {$G_{2}$};
        \filldraw[black] (2,1.5) circle (2pt) node[anchor=west]{$S_1 = S_3$};
        \filldraw[black] (2,3.5) circle (2pt) node[anchor=west]{$S_2$};
        \draw [decorate,decoration={brace,amplitude=5pt},xshift=0,yshift=0pt]
(-0.2,0.5) -- (-0.2,1.3) node [black,midway,xshift=-0.6cm]
{\footnotesize $\leq \epsilon$};
    \draw [decorate,decoration={brace,amplitude=5pt, mirror},xshift=0,yshift=0pt]
(2.1,1.7) -- (2.1,3.3) node [black,midway,xshift=0.6cm]
{\footnotesize $> \epsilon$};

    \end{tikzpicture}
    \caption{Visual representation of the setting in Example~\ref{ex:rec-not-trade} showing that when computed in the presence of noise, \tA{$\epsilon$} solutions may break down. The noisy version consists of updating $S_2$ to the location highlighted in red.}
    \label{fig:rec-not-trade-diag}
\end{figure}

By making the distance between $S_2$ and $S_1$ arbitrarily small, we can construct examples where even when the amount of noise is negligible compared to the allowed approximation factor, $S_1$ can still potentially be incorrectly classified as \tA{$\epsilon$}. 

We note that we can define a stricter notion of tradeoff approximation that guarantees a solution will be \tA{$\epsilon$} even if calculated with noisy utilities, but for the same reasons as demonstrated in Example~\ref{ex:sig-trade-problems}, such solutions may not always exist, making it difficult to find solutions that are guaranteed to be \tA{$\epsilon$} in a noisy setting.

\begin{lemma}\label{lem:noisy-trade}
Recall the notion of $(\epsilon_1, \epsilon_2)$-significant tradeoff leximax as presented in Definition~\ref{def:sig-trade}. Any $(\epsilon - 2\delta, 2\delta)$-significant tradeoff leximax solution when calculated using noisy utilities within an additive $\delta$ of their true values is guaranteed to be \tA{$\epsilon$} with respect to the true utilities. 
\end{lemma}

\begin{proof}
Let $u: \sols \times \grps \rightarrow \mathbb{R}_{\geq 0}$ be the true utilities, and $u_{\delta}: \sols \times \grps \rightarrow \mathbb{R}_{\geq 0}$ define noisy utilities such that $u_\delta(S, G) \in [u(S, G) - \delta, u(S, G) + \delta]$ for all $S \in \sols$ and $G \in \grps$.

Suppose we have some $S'$ and $i$ such that 
$$u(S', G_{[i]}) > u(S, G_{[i]}) + \epsilon$$

In the noisy setting, we are therefore guaranteed to have
$$u_{\delta}(S', G_{[i]}) > u_{\delta}(S, G_{[i]}) + \epsilon - 2\delta$$

By definition of $(\epsilon - 2\delta, 2\delta)$-significant tradeoff leximax, we therefore have some $j < i$ such that 

$$u_{\delta}(S', G_{[j]}) < u_{\delta}(S, G_{[j]}) - 2\delta$$

Switching back to non-noisy utilities, we are guaranteed that

$$u_{\delta}(S', G_{[j]}) < u_{\delta}(S, G_{[j]}) - 2\delta + 2\delta$$
$$u_{\delta}(S', G_{[j]}) < u_{\delta}(S, G_{[j]})$$

and therefore $S$ is \tA{$\epsilon$}. 
\end{proof}

Having considered how noise may affect \tA{$\epsilon$} approximations, we now turn to \sig{$\epsilon$} approximations. Here, we find that \sig{$\epsilon$} solutions are somewhat robust to noise, in that they satisfy a slightly relaxed definition of significance. 

First, we note that in Example~\ref{ex:rec-not-trade}, $S_1$ is also \sig{$\epsilon$} in the noisy setting, but not in the non-noisy setting, and so this example also demonstrates how the standard definition of \sig{$\epsilon$} may not be robust to noise. However, we can offer the following guarantee with respect to a modified notion:

\begin{lemma}\label{lem:sig-w-noise}
Say that a solution $S$ is \sig{$(\alpha_1, \alpha_2)$} if there exists some choice of slack $\vec{\beta} = (\beta_1, ..., \beta_m)$ with $\beta_i: \sols \rightarrow [\alpha_1, \alpha_2]$ such that $S \in \sols_m^{\beta}$, where $\sols_0^{\beta} = \sols$ and 
$$\sols_i^{\beta} = \{S \in \sols_{i - 1}^{\beta}: u(S, G_{[i]}) \geq \max_{S' \in \sols_{i - 1}^{\beta}}u(S', G_{[i]}) - \beta_i(S)\}.$$

Then, any \tA{$\epsilon$} solution calculated in the presence of $\delta$ additive noise is guaranteed to be \sig{$(\epsilon - 2d, \epsilon + 2d)$}.
\end{lemma}

\begin{proof}
We begin by proving a property about how the sorted vector of group utilities is affected by added noise. 

\begin{claim}
For any solution $S$ and $i \in [m]$, we have $|u(S, G_{[i]}) - u_{\delta}(S, G_{[i]})| \leq \delta$.
\end{claim}

\begin{claimproof}
Suppose for purposes of contradiction that for some $S$ and $i$, $|u(S, G_{[i]}) - u_{\delta}(S, G_{[i]})| > \delta$. We can assume without loss of generality that $u(S, G_{[i]}) > u_{\delta}(S, G_{[i]})$.

Let $G_j$ and $G_{j_\delta}$ be the groups used to calculate $u(S, G_{[i]})$ and $u_{\delta}(S, G_{[i]})$, respectively. If $G_j = G_{j_\delta}$, this is a contradiction because it means the noise on group $G_j$ was more than $\delta$. 

Otherwise, in order to ensure that the noise requirements $|u(S, G_j) - u_{\delta}(S, G_j)|, |u(S, G_{j_{\delta}}) - u_{\delta}(S, G_{j_{\delta}})| \leq \delta$, we must have that $u_{\delta}(S, G_{j}) > u_{\delta}(S, G_{[i]})$ and additionally that $u(S, G_{j_{\delta}}) < u(S, G_{[i]})$, otherwise these constraints on noise cannot be true. 
Because $G_{j_{\delta}}$ is below $G_j$ in the sorted groups vector according to $u$, but $G_{j}$ is above $G_{j_{\delta}}$ in the sorted groups vector according to $u_{\delta}$, but $G_j$ and $G_{j_{\delta}}$ occupy the same index in both sorted vectors, we must be able to find some other $G_k$ such that $u(S, G_k) > u(S, G_{[i]})$ but $u_{\delta}(S, G_k) < u_{\delta}(S, G_{[i]})$. 

However, this implies that 
$$u(S, G_k) > u(S, G_{[i]}) > u_{\delta}(S, G_{[i]}) > u_{\delta}(S, G_k)$$
$$u(S, G_k) > u(S, G_{[i]}) > u_{\delta}(S, G_{[i]}) > u(S, G_k) - \delta$$

and so we must have $|u(S, G_{[i]}) - u_{\delta}(S, G_{[i]})| \leq \delta$. This contradicts our original assumption, and so we conclude that for all $i$ and $S$, $|u(S, G_{[i]}) - u_{\delta}(S, G_{[i]})| \leq \delta$.
\end{claimproof}

With this claim in hand, we can act as if noise was applied with respect to the sorted vector of group utilities rather than the groups themselves. 

We define a new amount of allowable slack $\vec{\beta}$ as follows, where $\sols_1^{\epsilon}, ..., \sols_m^{\epsilon}$ are the recursively defined sets used to calculated \sig{$\epsilon$} on the noisy utilities. 

$$\beta_i(S) = \begin{cases}\epsilon + 2\delta & S \in \sols_i^{\epsilon} \\ \epsilon - 2\delta & \text{otherwise}\end{cases}$$

We proceed by induction, noting that $\sols_0^{\epsilon} = \sols_0^{\beta} = \sols$.

Suppose that for all $j < i$, we have $\sols_j^{\epsilon} = \sols_j^{\beta}$. 

Then, if $S \in \sols_i^{\epsilon}$, we have that
\begin{align*}
    u_{\delta}(S, G_{[i]}) &\geq \max_{S' \in \sols_{i-1}^{\epsilon}}u_{\delta}(S', G_{[i]}) - \epsilon\\
    u(S, G_{[i]}) + \delta &\geq \max_{S' \in \sols_{i-1}^{\epsilon}}u(S', G_{[i]}) - \epsilon - \delta\\
    u(S, G_{[i]})  &\geq \max_{S' \in \sols_{i-1}^{\epsilon}}u(S', G_{[i]}) - \epsilon - 2\delta\\
    u(S, G_{[i]})  &\geq \max_{S' \in \sols_{i-1}^{\beta}}u(S', G_{[i]}) - \beta_i(S)\\
\end{align*}

and so $S \in \sols_i^{\beta}$ as well. On the other hand, if $S \not\in \sols_i^{\epsilon}$, then let $j$ be the smallest $j$ such that $S \not\in \sols_j^{\epsilon}$. We must have

\begin{align*}
    u_{\delta}(S, G_{[j]})  &< \max_{S' \in \sols_{j-1}^{\epsilon}}u_{\delta}(S', G_{[j]}) - \epsilon\\
\end{align*}

In the non-noisy setting, we are therefore guaranteed that
\begin{align*}
u(S, G_{[j]}) -\delta  &< \max_{S' \in \sols_{j-1}^{\epsilon}}u(S', G_{[j]}) +\delta - \epsilon\\
u(S, G_{[j]}) &< \max_{S' \in \sols_{j-1}^{\epsilon}}u(S', G_{[j]}) - \epsilon + 2\delta\\
u(S, G_{[j]}) &< \max_{S' \in \sols_{j-1}^{\epsilon}}u(S', G_{[j]}) - \beta_j(S)\\
u(S, G_{[j]}) &< \max_{S' \in \sols_{j-1}^{\beta}}u(S', G_{[j]}) - \beta_j(S)\\
\end{align*}

And therefore $S \not\in \sols_{j}^{\beta} \supseteq \sols_i^{\beta}$ as well, so we can conclude that $\sols_i^{\epsilon} = \sols_i^{\beta}$ for all $i = 1, ..., m$, and hence any $S$ that is \sig{$\epsilon$} in the noisy setting must be \sig{$(\epsilon - 2\delta, \epsilon + 2\delta)$}.
\end{proof}

Thus, we conclude that while noisy \sig{$\epsilon$} solutions are not guaranteed to be \sig{$\epsilon$} with respect to the true utilities, they will still satisfy a slightly relaxed notion of significance that allows for slack to vary within an interval of size $4\delta$ around the constant $\epsilon$ slack used in standard significance. When $\delta$ is tiny compared to $\epsilon$, this is only a tiny change in the allowed slack values.

\section{Solutions via Linear Programming}\label{sec:lp-approaches}

As discussed in Section~\ref{sec:intro-algs}, we provide efficient algorithms for a particular natural choice of cohort selection setting. In particular, we consider modeling utility as the sum of the utilities that a subgroup draws from each individual member of the selected cohort, and rather than outputting a lexicographically maximal cohort, we output a lexicographically maximal vector of marginal selection probabilities that provides leximax utility in expectation. 

\subsection{Problem Setting}
We begin by discussing our choice of utility function and randomized selection approach in more detail.

\subsubsection{Linear Utility Function}

Let $\comm$ be a set of potential committee members of size $n$. We assume that each subgroup $G_j \in \grps$ has a value for each individual committee member $c_i \in \comm$, denoted by $v_{ij} \in [0, 1]$. 

When we choose our set of solutions to be $\comm^{(k)}$, the set of all cohorts of size $k$, these values can now be combined to give a group's utility for any particular cohort as the sum of its values for the cohort members. Given a cohort $C = \{c_1, ..., c_k\} \in \comm^{(k)}$ and subgroup $G_j \in \grps$, this utility function can be written formally as 

$$u(C, G_j) = \sum_{i = 1}^kv_{ij}$$

This linear utility can easily be extended to the randomized case. Assuming $\comm$ has size $n$, any vector of individual assignment probabilities $D = \{x_1, \dots, x_n\} \in \mathcal{D} := [0, 1]^n$, called \emph{marginal (selection) probabilities}, provides an approach to randomly selecting a cohort of candidates from $\comm$ where each $c_i$ is included in the cohort with probability $x_i$, independent of the other candidates. The expected utility of a particular group $G_j$ over a distribution $D \in [0, 1]^n$ is then 
$$u(D, G_j) = \sum_{i = 1}^n x_i v_{ij}.$$
We will restrict our search to marginal distributions that output a cohort with expected size $k$ ($\sum_{i = 1}^n x_i = k$).

\subsubsection{Randomized Selection Approach}

Our algorithms output a vector of marginal selection probabilities $D = \{x_1, ..., x_n\} \in \dist := [0, 1]^n$, such that when each cohort member $c_i$ is independently included in the cohort with probability $x_i$, we get a cohort of size $k$ in expectation such that the vector of expected utilities is lexicographically maximal. This independent sampling procedure provides a simple way to randomly select a cohort.

This distributional approach to selection renders the problem tractable, while we show in Section~\ref{sec:integer-hardness} that finding deterministic leximax solutions is NP-hard. Moreover, it provides a fair way to get around the issue that in a deterministic setting, there may be multiple leximax cohorts that each favor a different subgroup. 

It's worth noting that this approach to cohort selection only gives a cohort with \emph{expected} size $k$. While such a selection procedure may be fine in situations where the desired size of the final cohort is somewhat flexible, sometimes it may be critical to get a cohort of size exactly $k$. In Section~\ref{sec:rounding}, we discuss a dependent rounding scheme that can be used to sample a cohort of size exactly $k$ with utilities that are still leximax in expectation.

In general, cohorts sampled from arbitrary leximax distributions are not guaranteed to provide groups with utility near their expected value. However, our choice of selection procedure guarantees that groups receive near-expected utility with high probability.

\begin{lemma}\label{lem:concentration}
Consider an arbitrary group $G_j$ and a lexicographically maximal vector of marginal selection probabilities $D \in \dist$ (with respect to the linear utility function defined above and with expected size $k$). 

Then, for any $\delta > 0$, we have 
$$\Pr_{C \sim D}[U(C, G_j) < U(D, G_j) - \delta] < e^{-2\delta^2/n}$$

(Where $n := |\comm|$ is the number of potential cohort members.)
\end{lemma}

\begin{proof}
Define random variables $X_1, ..., X_n$ such that $X_i$ is $v_{ij}$ if $c_i$ is included in the cohort, and zero otherwise. 

According to our random selection procedure, these are all independent variables with each $X_i$ taking on the value $v_{ij}$ with probability $x_i$.

Because $v_{ij} \in [0, 1]$ for all $i \in [n]$ by definition, this is the sum of $n$ independent random variables with values bounded between 0 and 1. 

Let $\mu:= \mathbb{E}[\sum_{i = 1}^nX_i]$. Applying an additive Chernoff bound~\cite{dubhashi_panconesi_2009} gives an upper bound on the probability that the sum of $X_i$s falls far below $\mu$:

$$\Pr[\sum_{i = 1}^nX_i < \mu - \delta] < e^{-2\delta^2/n}$$

For any $\delta > 0$. Thus, because $u(C, G_j) \sim \sum_{i = 1}^nX_i$, and $u(D, G_j) = \mathbb{E}[\sum_{i = 1}^nX_i]$, we get the statement of the lemma.
\end{proof}

To contextualize this result, consider some group $G_j$ that is expected to get about half of their maximum possible utility for a leximax solution when $k = 50$ and $n = 100$. Because the values for each individual are defined to be between 0 and 1, this means that $G_j$ has an expected utility of 25. Then, Lemma~\ref{lem:concentration} implies that they are guaranteed to get at least half their expected utility more than 95\% of the time. In comparison, an arbitrary leximax distribution can potentially only guarantee that $G_j$ gets more than half their expected utility with probability 1/3. These concentration guarantees also hold for cohorts of size exactly $k$ outputted by our suggested rounding approach. More details can be found in Section~\ref{sec:rounding}.

Having explained and justified our choice of utility function as well as randomized selection approach, we now present our algorithms that calculate exact and approximate leximax solutions in this setting.

\subsection{Leximax distribution over committee members}
To find a marginal distribution over each potential committee member in $\comm$, we break up the problem into multiple, recursively-defined subproblems to uncover the ranking of subgroup utilities in the leximax optimal solution as well as their optimal values. 

Balan et. al \cite{balan2008} approach this problem by reducing the domain of solutions in each level of optimization. They choose the $(i+1)$-th subgroup to be the subgroup that least-constrains the domain of potential leximax solutions. Overall, their approach finds a leximax-optimal marginal distribution over potential committee members that requires $O(|\grps|)$ calls to a linear program at each of the $|\grps|$ iterations, giving us $O(|\grps|^2)$ total calls. However, this approach of limiting the domain of the possible solutions requires fixing an order of worst off groups in every iteration. When approximate notions of leximax are introduced, there can be multiple possible orderings of groups to consider.

We suggest finding the leximax distribution over individuals as a series of linear programs with a linear number of variables and a number of constraints that increases from linear to exponential as the series progresses. In the first LP, we are finding the maxmin utility $\gamma_1$ using the values $v_{ij}$ that each group has for individual cohort candidates: 
\begin{align*}
    \begin{array}{ll}
        \mbox{maximize}_{x, \gamma_1}  &  \gamma_1 \\
        \mbox{subject to} & \sum_{i=1}^{n}x_i = k \\
                        & 0 \le x_i \le 1 \\
                        & \sum_{i=1}^n x_i v_{ij} \ge \gamma_1 \quad j = 1, \dots, m 
    \end{array}
\end{align*}

Once the optimal lower bound for the worst off group, $\gamma_1^*$ is found, is found, we can then maximize the utility of the second-worse-off-group. Ogryczak et al. \cite{ogryczak2005telecommunications} observed that maximizing the $\mathbf{\gamma}=(\gamma_1, \dots, \gamma_m)$ vector is equivalent to maximizing for the cumulative sum of $\gamma_i$'s from $i=1, \dots, m$. Thus, to find the leximax distribution of individuals, we optimize a series of $m$ linear programs using the cumulative leximax values as a constraint. The $m$-th last LP will be as follows:
\begin{align*}
    \begin{array}{ll}
        \mbox{maximize}_{x, \gamma_m}  &  \gamma_m\\
        \mbox{subject to} & \sum_{i=1}^{n}x_i = k \\
                        & 0 \le x_i \le 1 \\
                        & \sum_{i=1}^n \sum_{G_j \in S} v_{ij}x_i \ge \sum_{s=1}^l \gamma_s^* \quad \forall l = 1, \dots, m, \forall S \subseteq \mathcal{G} \ s.t. \ |S| = l
    \end{array}
\end{align*}
Since we must ensure that the sum of utilities is above the minimum utility for all subgroups, the last constraint requires that the sum of utilities over all sized-$l$ subsets of groups be greater than the sum of the $l$ optimal $\gamma^*$-s (i.e. $\sum_{i=1}^l \gamma_i^*$) from previous iterations. This creates $\binom{m}{l}$ constraints for the $l$-th LP. Algorithm \ref{alg:lp_cand} describes the iterative process of finding a leximax distribution where in each successive problem we add additional constraints on the minimum value of the sum of utilities. In our setting of linear utilities, we can solve each linear program in polynomial time with the ellipsoid method using a polynomial-time separation oracle.

\begin{algorithm}
\caption{\textsc{leximaxCandidates} Finding the leximax distribution over candidates}%
\KwIn{$v \in \R_{\ge0}^{n\times m}$ values of each group for each candidate.}
\KwOut{$\{x_1, \dots, x_n \}$ leximax distribution over candidates.}
Constraints = \{${\sum_{i=1}^n x_i v_{ij} \ge \gamma_1 \  j = 1, \dots, m};\ 0 \le x_i \le 1;\ \sum_{i=1}^{n}x_i = k \}$\; 
$\gamma_1^* \gets \max_{x, \gamma_1} \gamma_1$ s.t. Constraints \; 
\For{$l \in 2, \dots ,m$}{
Constraints = Constraints $\cup \{\sum_{i=1}^n \sum_{G_j \in S} v_{ij}x_i \ge \sum_{s=1}^l \gamma_s^* \ \forall S \subseteq \mathcal{G} \ s.t. \ |S| = l$ \}\;
$\gamma_i^* \gets \max_{x, \gamma_i} \gamma_i$ s.t. Constraints given  $\gamma_1^* , \dots, \gamma_{i-1}^*$(previously computed)\; 
}
\label{alg:lp_cand}
\end{algorithm}

\begin{lemma}
For $n$ candidates and $m$ groups, the running time of Algorithm \ref{alg:lp_cand} is polynomial in $n$ and $m$.   
\label{lem:poly-lp}
\end{lemma}
\begin{proof}
We run $m$ LPs in total. For each LP, the running time is the number of steps the ellipsoid algorithm takes multiplied by the time per iteration. 
For an efficient implementation, the ellipsoid algorithm needs (1) a feasible initial solution and (2) a polynomial-time separation oracle.

(1) For an initially feasible solution, set $\gamma_1 = 0$ and all $x_i = k/n$ for the first LP. It is easy to check that this gives a feasible solution. In the $m$-th LP use the $x$ values of the solution to the previous LP and $\gamma_m =0$ as the initial solution. This solution is feasible as all but the last constraint are identical to the previous LP and, thus, the $x$-values of the previous solution fulfill them. For the last constraint, note that the right side of the inequality equals the next-to-last constraint. As all utility values are non-negative, summing over a larger set $G$ on the left side only increases the value of the left side in comparison to the value of the next-to-last constraint. Thus, the last constraint is fulfilled as well for $\gamma_m = 0$.

(2) Given a vector of $x$-values and a vector of minimum utilities $\gamma_i$ the goal of a separation oracle is to decide whether these values fulfill the LP and, if they do not, find a constraint that is violated by them.
The time of the separation oracle  dominates the running time per iteration of the ellipsoid algorithm. Thus, it suffices to give a polynomial-time separation oracle. 
We present our separation oracle in Algorithm \ref{alg:sep-oracle}.
It first checks whether all $x$-values fall into the correct range and add up to $k$.
Then it computes the utility $y_j$ of each subgroup $G_j$ and sorts them in non-decreasing order of $y$-value.
Instead of checking all $\binom{m}{l}$ constraints for each set of $l$ subgroups, it
uses the following observation:
it suffices to check that, for each $l$, the sum of the utilities of the $l$ groups with \emph{smallest} utilities is at least $\sum_{s=1}^l \gamma_s$. The reason is that every other set of $l$  subgroups must have cumulative utility at least as large. If, however, the set of $l$ subgroups with minimum utility does not have high enough cumulative utility, then a violating constraint has been found.

Summing up utilities across $n$ candidates takes $O(n)$ time, sorting the resulting utility vector $y$ takes $O(m \log m)$ time. In total, this separation oracle checks if all the constraints are satisfied in $O(m \log m + n)$ time.
\begin{algorithm}
\caption{\textsc{Separation Oracle} Checking if a constraint has been violated by a given solution $x$ and $\gamma$}%
\KwIn{$v \in \R^{n \times m}_{\ge0}$, values of each group for each candidate, $\{x_1 , \dots , x_n \}$ candidate solution, $\{\gamma_1 , \dots , \gamma_l \}$ minimum utilities for the $l$-th LP}
\KwOut{\{TRUE or a violated constraint\}}
$S \gets 0$\; 
\For{$i = 1, \dots, n$}{
    \If{$x_i > 1$ or $x_i < 0$}{
    \Return $\{0 \le x_i \le 1\}$} 
    $S \gets S + x_i$ ;\
}
\If{$S \neq k$}{
    \Return $\{\sum_{i=1}^n x_i \le k\}$} 
    
$ y_j \gets \sum_{i=1}^n v_{ij} x_i \quad \forall j=1 , \dots,  m $\;  
$ \tilde{y} \gets \textsc{SORT}(y)$\; 
$U_{min} \gets 0 $\; 
\For{$l = 1 , \dots, m$}{
$U_{min} \gets U_{min} +  \tilde{y}_l$\;
\If{$U_{min} < \sum_{s=1}^l \gamma_s$}{
    \Return FALSE as this constraint does not hold: $\{\sum_{i=1}^n \sum_{G_j \in S} v_{ij}x_i \ge \sum_{s=1}^l \gamma_s \quad \forall S \subseteq \mathcal{G} \ s.t. \ |S| = l$\}} 
}
\Return TRUE
\label{alg:sep-oracle}
\end{algorithm}

For the ellipsoid method, we are guaranteed convergence in $k$ steps where $k \le 2 n^2 \log(\frac{R}{r})$ where $R$ is the initial radius and $r$ is the final radius of  the feasible region \cite{bland1981ellipsoid}. For our feasibility region, $R$ is exponential with respect to the input size (i.e. $O(2^n)$) which means $\log(\frac{R}{r})$ is linear with respect to $n$.  Since the separation oracle and centroid method at each step runs in polynomial time and there are at most $\tilde{O}(n^2)$ steps,  Algorithm \ref{alg:lp_cand} also runs in polynomial time. 
\end{proof}

\subsubsection{Approximate Leximax distribution over candidates}
When finding approximate leximax distributions over candidates, the approach of Balan et. al. \cite{balan2008} can no longer be applied since choosing the subgroup that least constrains the domain of potential solutions may yield multiple subgroups when the leximax objective is approximate. Thus, there is no single ordering of worst-off-groups to rely on when considering group utility. %
However, we can easily modify our recursive linear program (Algorithm \ref{alg:lp_cand}) to find an 
an \rA{$\epsilon$} solution (Definition \ref{def:cg-recursive-approx}) for a given `slack' vector $\vec{\alpha} = (\alpha_1, \dots, \alpha_m)$.
While the first LP is the same as the exact case, we can loosen the constraints in the m-th LP as follows: 
\begin{align*}
    \begin{array}{ll}
        \mbox{maximize}_{x, \gamma_m}  &  \gamma_m\\
        \mbox{subject to} & \sum_{i=1}^{n}x_i = k \\
                        & 0 \le x_i \le 1 \\
                        & \sum_{i=1}^n \sum_{G_j \in S} v_{ij}x_i > \sum_{s=1}^l (\gamma_s^* - \alpha_s) \ \forall S \subseteq \mathcal{G} \ s.t. \ |S| = l, \ l = 1, \dots, m-1 \\
                        & \sum_{i=1}^n \sum_{G_j \in \mathcal{G}} v_{ij}x_i > \sum_{s=1}^{m-1} (\gamma_s^* - \alpha_s) + \gamma_m 
    \end{array}
\end{align*}

For a \sig{$\epsilon$} approximate solution, we can set all the $\alpha_i$'s equal to $\epsilon$ and apply algorithm \ref{alg:lp_cand} with modified constraints as described above.

\subsection{Rounding Distributions Over Candidates}
\label{sec:rounding}
Once we obtain a distribution over cohort candidates from Algorithm \ref{alg:lp_cand}, we can sample each individual $i$ with probability $x_i$ independently. The total size of the committee follows a Poisson Binomial distribution which will be size-$k$ in expectation where $k = \sum_{i=1}^n x_i$ according to our constraints.

If a committee of size $k$ is a hard constraint, we can instead take a rounding approach similar to previous work in cohort selection \cite{bairaktari2021fair}. For finite samples in our cohort selection setting, we can employ a dependent rounding scheme that guarantees that the utilities for each subgroup is leximax in expectation while the size of the cohort is exactly $k$~\cite{srinivasan2001distributions}.

The rounding scheme described in~\cite{srinivasan2001distributions} results in a distribution over cohorts of size exactly $k$ such that the marginal inclusion probability for each potential cohort member is still satisfied, giving us the leximin utility values in expectation when the utility function is linear over cohort members. The scheme has the added benefit that the events corresponding to the inclusion/exclusion of each cohort member are negatively correlated. Because Chernoff bounds such as the one used in our proof of Lemma~\ref{lem:concentration} have been shown to also hold in settings where random variables are not independent but are negatively correlated (See~\cite{doerrconcentrationbook}, Theorem 1.10.24), our concentration guarantees also apply to solutions outputted by the rounding scheme.

\subsection{Integer Solution}\label{sec:integer-hardness}
Although the focus of this work is providing distributions over candidates and cohorts, we also touch briefly on the problem of finding integer leximax cohorts. An exact integer leximax solution removes the randomness inherent in rounding from a distributional solution. However, we show such an integer solution is NP-hard to find. Moreover, the weaker maxmin version of the problem (see below) is NP-hard to compute. 

Given a set of candidates $\comm = \{c_1, \dots, c_n\}$, a set of groups $\grps  = \{G_1, \dots, G_m\}$, and the values $v \in \R^{n \times m}_{\ge 0}$ of each group for each candidate such that the utility of a group for a cohort is its average value over the cohort's candidates, the integer leximax cohort selection problem is: 
\begin{align*}
    \begin{array}{ll}
        \mbox{maximize}  &  \gamma_1, \dots, \gamma_m\\
        \mbox{subject to} & \sum_{i=1}^{n}x_i = k \\
                        &  x_i \in \{0, 1\} \\
                        & \sum_{i=1}^n \sum_{G_j \in G} v_{ij}x_i \ge \sum_{s=1}^l \gamma_s \quad \forall G \subseteq \mathcal{G} \ s.t. \ |G| = l, l = 1, \dots, m 
    \end{array}
\end{align*}
The simpler integer \emph{maxmin cohort  selection problem with cardinality $k$} determines  a set of candidates defined by $x_i$'s such that the minimum utility  of any group is maximized: 
\begin{align*}
    \begin{array}{ll}
        \mbox{maximize}  &  \gamma\\
        \mbox{subject to} & \sum_{i=1}^{n}x_i = k \\
                        &  x_i \in \{0, 1\} \\
                        & \sum_{i=1}^n x_i v_{ij} \ge \gamma  \quad \forall j = 1, \dots, m
    \end{array}
\end{align*}

Next we show the hardness of the maximin cohort selection problem and even of the following \emph{integer $\epsilon$-approximate maxmin cohort selection problem with cardinality $k$}, where $0 \le \epsilon$ is a constant: Determine  a set of candidates defined by $x_i$'s such that the minimum utility of any group is within an additive error of $\epsilon$ of $\gamma$, the maximum minimum utility possible.

\begin{lemma}
For $\epsilon < 0.5$
the integer $\epsilon$-approximate maxmin cohort selection problem is NP-hard. It is also NP-hard to determine the number of groups with non-minimum utility to within a factor of $(e-1)/e + o(1)$.
\label{lem:hardness}
\end{lemma}

\begin{proof}
We will show that finding a solution to the maxmin cohort selection problem solves the NP-hard problem of Minimum Hitting Set. Thus, the maximin cohort selection problem is also NP-hard. 

The Minimum Hitting Set Problem is defined as follows. Let $C$ be a collection of subsets of a finite set $S$. A hitting set for $C$ is a subset $S' \subseteq S$ such that $S'$ contains at least one element from each subset in $C$.

Given a hitting set problem
we construct  a maxmin cohort selection problem as follows. Given $s_i$ and $C_j$, we set $v_{ij}$ such that $v_{ij} = 1$ if $s_i \in C_j$ and 0 otherwise. Then for each $k$ starting with $k=1$ to $k$, we solve the   $\epsilon$-approximate integer leximax cohort problem with cardinality $k$.
As $1-\epsilon > \epsilon$ for $\epsilon < 0.5$ a return value 
of at least $1 - \epsilon$ implies that the minimum utility is at least 1, while a return value of at most $\epsilon$ implies that the minimum utility is 0.
Thus, a hitting set of size $k$ exists iff the $\epsilon$-approximate integer maxmin cohort selection problem for a cohort of size $k$ returns a value of at least $1 - \epsilon$.

Hence, the smallest value of $k$ such that 
the return value for the problem with cardinality $k$ is at least $1-\epsilon$ gives us the size of the minimum hitting set. The set of indices $i$ such that $x_i = 1$ gives the elements of the hitting set.
 As the reduction used at most $k$ calls to the maxmin cohort selection problem to solve the  Minimum Hitting Set Problem, the $\epsilon$-approximate integer maxmin cohort selection problem must also be NP-hard. 

Furthermore, given a parameter $k$ which limits the size of the hitting set, the \emph{maximum coverage version} of the problem asks for the maximum number of sets covered by a hitting set of size $k$. 
It is NP-hard to approximate this number  within $(e-1)/e + o(1)$~\cite{feige1998threshold}. It follows that it is NP-hard to approximate within this factor how many groups can have non-minimum utility if at most $k$ candidates are selected.

\end{proof}

\section{Discussion and Future Work}
Motivated by the problem of selecting representative cohorts, we turned to a lexicographically maximal definition of optimal representation. We investigated existing approximations of leximax fairness and introduced new definitions which consider semantic notions of noise and tradeoffs. In settings where utilities or objectives are roughly estimated and leximax fairness is desirable, the approximate notions of leximax in this paper may be useful as alternatives to exact leximax.  

While we gave a polynomial time algorithm which computes a leximax distribution over a pool of candidates that is effective for both exact and approximate notions of leximax, finding an algorithm for approximation notions of leximax that is more efficient than exact algorithms remains an open problem. Furthermore, our setting of linear utilities is a natural assumption but can be extended to sub-modular or other classes of utility functions. 

In another direction, our approximation notions all reason about allowing for additive amounts of error. However, considering what notions, especially those in line with \sig{$\epsilon$}, might arise from multiplicative error could be a useful direction to explore.

Finally, we only considered how the presence of additive noise might affect our definitions, but other models of noise specific to different domains may also be considered. Noise can appear not just based on entire cohorts or distributions but also for candidates individually. Modeling how noise from individual candidates accumulate over over cohorts and distributions of candidates will vary depending on the utility function but is a promising direction to explore.

\bibliography{lipics-v2021-sample-article}

\appendix

\end{document}